\documentclass[letterpaper,10pt,conference]{ieeeconf}
\IEEEoverridecommandlockouts
\usepackage{tabularx}
\usepackage{longtable}
\usepackage{amsmath,amssymb}
\usepackage{graphicx}
\usepackage{cite}
\usepackage{url}
\usepackage{algorithm}
\usepackage{algpseudocode}
\usepackage{booktabs}
\usepackage{multirow}
\usepackage{float}
\newtheorem{theorem}{Theorem}
\newtheorem{remark}{Remark}
\newtheorem{corollary}{Corollary}

\begin{document}

\title{Forward-Invariant Policy Classes for Safe Reinforcement Learning in Multicopter Control}

\author{%
Chieh~Tsai$^{1}$,
Muhammad~Junayed~Hasan~Zahed$^{2}$,
Jinzhi~Shen$^{1}$,
Yi~Xie$^{1}$,
Ruoshan~Lan$^{3}$,\\[1mm]
Majed~Obaid$^{4}$,
Salim~Hariri$^{1}$,
and~Hossein~Rastgoftar$^{1,2}$%
\thanks{$^{1}$Department of Electrical and Computer Engineering,
University of Arizona, Tucson, AZ 85719, USA
(e-mail: \{hariri, jinzhis9, yix\}@arizona.edu).}%
\thanks{$^{2}$Department of Aerospace and Mechanical Engineering,
University of Arizona, Tucson, AZ 85719, USA
(e-mail: \{mjhz, hrastgoftar\}@arizona.edu).}%
\thanks{$^{3}$Department of Computer Science,
University of Arizona, Tucson, AZ 85719, USA
(e-mail: ruoshanlan@arizona.edu).}%
\thanks{$^{4}$Community Medicine and Primary Care Department,
College of Medicine, Umm Al-Qura University.}%
}
\maketitle

\begin{abstract}
This paper proposes a reinforcement learning (RL) framework for
safe gain scheduling based on forward-invariance-induced action-space
design. Under stated nominal-model and inversion-domain assumptions,
rather than enforcing safety through runtime shielding
or penalty-based constraints, safety is embedded directly into
the policy class. Specifically, we construct a finite library of
feedback controllers sharing a common Lyapunov certificate
that establishes forward invariance of a prescribed admissible
set under arbitrary switching. Consequently, any policy whose
actions are restricted to this library, including policies
encountered during RL exploration, inherits the same
certificate. Policy optimization can therefore focus on
closed-loop performance without runtime safety filtering or
action projection. The framework is instantiated for
quadcopter hover regulation, where a DQN schedules among
certified feedback controllers. Nonlinear MuJoCo simulations
demonstrate state-dependent gain scheduling and empirically
evaluate robustness to wind, model mismatch, sensor noise,
and sensing delay. The results illustrate how safety
certification can be separated from policy optimization by
learning over a forward-invariant policy class.
\end{abstract}

\section{Introduction}

Reinforcement learning (RL) offers a promising framework for adaptive control of nonlinear dynamical systems under changing operating conditions. However, in safety-critical systems such as quadcopters, autonomous vehicles, and robotic platforms, learning must improve performance without violating hard safety constraints during training or deployment. Existing methods commonly enforce safety through penalties, shields, barrier constraints, or online filters, treating safety as an external correction rather than an intrinsic property of the policy class. 
This challenge is particularly important for quadcopters, where nonlinear dynamics, underactuation, actuator limits, and attitude constraints complicate adaptive control across operating regimes. Model-based controllers can provide strong stability and tracking guarantees, but their gains are typically fixed offline. A single gain configuration must therefore balance rapid transient rejection, precise hover regulation, and disturbance robustness, often resulting in conservative performance.

This paper develops a safe RL framework that embeds safety directly into the action space. The closed-loop system is formulated as an MDP with a finite set of pre-certified stabilizing feedback laws, each preserving forward invariance of a prescribed safe set. Thus, every policy over the admissible action space is safety-preserving by construction, eliminating the need for external safety corrections. The framework is demonstrated for quadcopter hover regulation. Figure~\ref{fig:framework_overview} summarizes the workflow: offline certification constructs the admissible controller library and supports policy training, while online DQN decisions are restricted to this certified library.

\subsection{Related Work}

Quadcopter trajectory tracking has been extensively studied using nonlinear, model-based control. Early works established modeling and low-level control through PID/LQ, backstepping, and experimental platforms \cite{Bouabdallah2004ICRA,Bouabdallah2004IROS,Castillo2004TCST,Madani2006IROS}, while subsequent studies addressed underactuation, actuator dynamics, and full-pose stabilization \cite{Hoffmann2007AIAA,Bouabdallah2007IROS,Nagaty2013JIRS}. More recent geometric and differential-flatness methods provide strong tracking guarantees and enable aggressive flight while avoiding Euler-angle singularities \cite{Lee2010CDC,Mellinger2011ICRA,Faessler2018RAL}. Despite their strong performance, these methods typically use gains fixed offline, which can be conservative across large transients, aggressive maneuvers, and near-hover operation.

\begin{figure}[!t]
  \centering
  \includegraphics[width=\columnwidth,trim=0 5bp 105bp 5bp,clip]{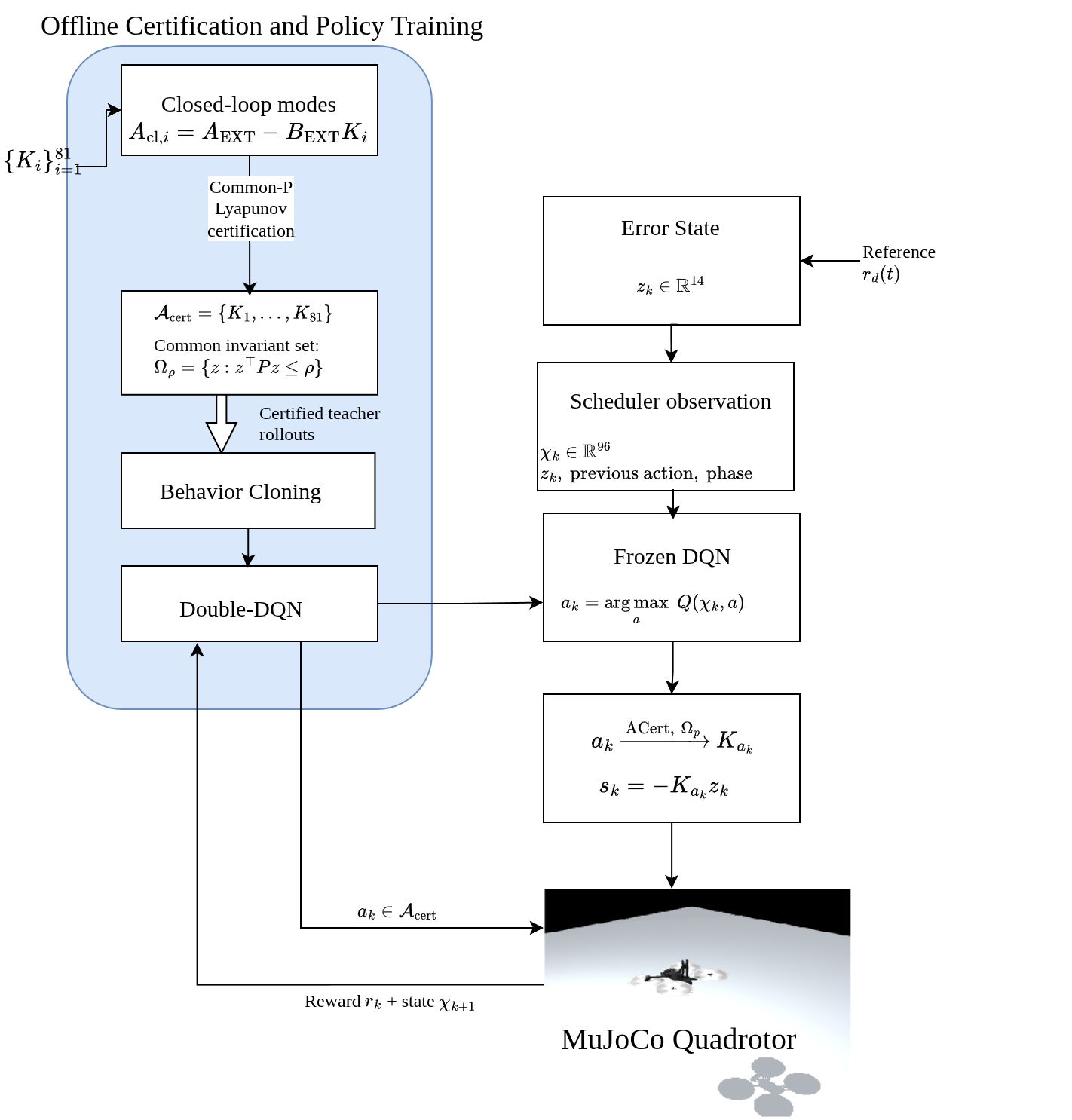}
  \caption{Offline common-Lyapunov certification and policy training (left), followed by online execution of the frozen DQN scheduler on the MuJoCo quadrotor (right).}
  \label{fig:framework_overview}
\end{figure}

Gain scheduling adapts controller aggressiveness across changing operating conditions. For quadrotors, gain-scheduled PID methods have improved fault-tolerant control and path tracking under actuator degradation and regime changes \cite{Milhim2010AIAA,Amoozgar2012IFAC}. However, existing approaches often rely on heuristic interpolation, fuzzy supervision, fault logic, or manually tuned switching rules \cite{Milhim2010AIAA,Amoozgar2012IFAC}, generally without formal guarantees that online gain updates preserve the safety envelope of the nonlinear closed-loop system.

This issue is closely tied to forward invariance. In constrained nonlinear
control, forward invariance requires that trajectories initialized in an
admissible set remain in that set for all future time
\cite{Blanchini1999Automatica}. Barrier-certificate methods established an
important verification framework for safety with respect to unsafe sets
\cite{Prajna2007TAC}, while control barrier function formulations provided
constructive real-time tools for enforcing invariant safe sets online
\cite{Ames2017TAC}. Related safe-learning work has also emphasized certified
regions of attraction, safe set expansion, and model-based stability
guarantees for uncertain nonlinear systems
\cite{Berkenkamp2016CDC,Berkenkamp2017NeurIPS}. For quadcopters, this
viewpoint is particularly important because position, attitude, angular-rate,
and thrust constraints are not merely soft performance targets; violating them
can compromise feasibility and invalidate controller assumptions. These
limitations motivate the use of admissible gain sets whose elements are certified
to preserve forward invariance of the closed-loop quadcopter dynamics.

From the RL viewpoint, sequential decision making is naturally formulated as
an MDP for state evolution, action selection, and long-horizon optimization
\cite{Bellman1957,Puterman1994,SuttonBarto2018}. Safe RL typically augments
this formulation with safety constraints
\cite{Achiam2017ICML,Chow2018NeurIPS}, often enforced through runtime
projection, penalties, Lyapunov constraints, or barrier filtering. In contrast,
we restrict the action space itself to safety-certified controllers.
Learning-based flight control can reduce manual tuning and compensate for modeling errors. RL has shown promising performance for UAV attitude control \cite{Koch2019TCPS}, while deep Q-learning established value-based learning over discrete action spaces \cite{Mnih2015Nature}. Safe RL has introduced constrained, Lyapunov-based, stability-certified, and barrier-function approaches \cite{Achiam2017ICML,Berkenkamp2017NeurIPS,Chow2018NeurIPS,Cheng2019AAAI}. However, most RL-based quadrotor controllers learn control actions or adapt controller parameters without embedding formal invariance guarantees into the gain-selection mechanism. In contrast, this work formulates safe gain scheduling as an MDP with a safety-certified action space. Rather than learning control commands directly, the DQN selects gain vectors from a finite library of pre-certified stabilizing controllers. Since each action preserves forward invariance of the prescribed safe set, every policy explored during training and deployment is safety-preserving by construction. This eliminates runtime safety correction and provides an interpretable, verification-friendly framework for adaptive quadcopter control.

\subsection{Contributions}

We consider the problem of safe learning for dynamical systems with uncertain models through control-theoretic invariance. By modeling the system evolution as a Markov decision process (MDP), safety is embedded directly into the decision-making structure rather than enforced through auxiliary mechanisms. In contrast to existing approaches, the proposed framework establishes safety at the level of the action space, yielding the following contributions:

\textbf{Contribution 1: Forward-Invariant Action Design for Safe Learning.} We construct a finite action space of feedback laws sharing a common Lyapunov certificate. The resulting certificate establishes forward invariance of a prescribed admissible set under arbitrary switching among the actions. Consequently, safety is independent of the action-selection policy and holds uniformly during both exploration and deployment.

\textbf{Contribution 2: Policy Optimization over an Invariant Policy Class.} Given the certified action space, policy optimization is
performed entirely within a forward-invariant policy class. Thus, the learning algorithm affects performance but not the underlying safety certificate. No runtime safety filter, action projection, dwell-time constraint, or safety penalty is required to preserve the certified set under the nominal closed-loop model.

\textbf{Contribution 3: Safety-Certified Gain Scheduling for Quadcopter Control.} We instantiate the framework for quadcopter hover control,
where a DQN schedules among a finite library of feedback controllers sharing a common invariance certificate. Nonlinear MuJoCo simulations evaluate nominal performance and empirically examine robustness to wind, model mismatch, sensor noise, and sensing delay.

\section{Problem Statement}
\label{Problem Statement}

We consider safe learning for a discrete-time dynamical system in
which safety is guaranteed by design, independently of the learning
process. To this end, we formulate an MDP defined by the tuple
\[
\mathcal{M}
=
(\mathcal{X},\mathcal{A},F,r,\gamma),
\]
where $\mathcal{X}\subseteq\mathbb{R}^n$ is a continuous state
space, $\mathcal{A}$ is a finite action set,
$r:\mathcal{X}\times\mathcal{A}\to\mathbb{R}$ is the stage reward,
and $\gamma\in(0,1)$ is the discount factor. Let
$\boldsymbol{\xi}_k\in\mathcal{X}$ denote the MDP state at decision
instant $k$. The system evolves according to
\begin{equation}
\boldsymbol{\xi}_{k+1}
=
F(\boldsymbol{\xi}_k,a_k),
\qquad
a_k\in\mathcal{A},
\label{Discrete}
\end{equation}
where
$F:\mathcal{X}\times\mathcal{A}\to\mathcal{X}$ is the closed-loop
transition map.

Rather than enforcing safety through online constraint handling or
action filtering, we embed safety directly into the admissible
action set.

\noindent\textbf{Problem 1 (Invariance-Induced Action Space Design).}
Construct a feedback parameterization and a finite action set
$\mathcal{A}$ such that
\begin{equation}
F(\boldsymbol{\xi},a)\in\mathcal{X},
\qquad
\forall
(\boldsymbol{\xi},a)\in
\mathcal{X}\times\mathcal{A}.
\label{eq:invariance_action}
\end{equation}
Then $\mathcal{X}$ is forward invariant under every admissible
action sequence, i.e.,
\begin{equation}
\boldsymbol{\xi}_0\in\mathcal{X}
\quad\Longrightarrow\quad
\boldsymbol{\xi}_k\in\mathcal{X},
\qquad
\forall k\geq 0.
\label{eq:discrete_invariance}
\end{equation}

\noindent\textbf{Problem 2 (Learning over an Invariant Policy Class).}
Given the action set $\mathcal{A}$ from Problem~1, determine a policy
\[
\pi:\mathcal{X}\to\mathcal{A}
\]
that maximizes
\begin{equation}
J_\pi(\boldsymbol{\xi}_0)
=
\mathbb{E}
\left[
\sum_{k=0}^{\infty}
\gamma^k
r\bigl(
\boldsymbol{\xi}_k,
\pi(\boldsymbol{\xi}_k)
\bigr)
\right].
\label{eq:policy_objective}
\end{equation}
Since every admissible action preserves $\mathcal{X}$, learning
reduces to policy optimization over a forward-invariant policy
class.

As a case study, we consider gain-scheduled quadcopter hover
regulation. The physical quadcopter state is denoted by
$\mathbf{x}$, while the MDP state $\boldsymbol{\xi}_k$ is
instantiated by the sampled tracking-error state
$\mathbf{z}_k=\mathbf{z}(t_k)$. Each action selects a feedback
controller from a finite invariance-certified library.
\section{Case Study: Quadcopter Hover Achievement}
\label{Case Study: Quadcopter Hover Regulation}

% We instantiate the proposed framework on quadcopter hover regulation. The objective is to drive the vehicle to a desired hover equilibrium while preserving safety. We proceed in two steps: first, we construct an invariance-certified action set; second, we derive a discrete-time transition model for learning.

% \subsection{Forward-Invariance-Preserving Control Design}

We consider the control-affine quadcopter dynamics
\begin{equation}
\dot{\mathbf{x}}=\mathbf{f}_0(\mathbf{x})+\mathbf{G}_0\mathbf{u},
\end{equation}
where
$
\mathbf{x}
=
\begin{bmatrix}
\mathbf{r}^\top &
\mathbf{v}^\top &
\boldsymbol{\eta}^\top &
\dot{\boldsymbol{\eta}}^\top &
T &
\dot T
\end{bmatrix}^\top
\in\mathbb{R}^{14}$ is the state and 

\[
\mathbf{u}
=
\begin{bmatrix}
u_T & u_\phi & u_\theta & u_\psi
\end{bmatrix}^\top
\in\mathbb{R}^{4}.
\]
is the input. Here, $\mathbf{r}\in\mathbb{R}^3$ and $\mathbf{v}\in\mathbb{R}^3$ denote position and velocity, $\boldsymbol{\eta}=[\phi,\theta,\psi]^\top$ is the Euler-angle vector, and $T$ and $\dot T$ are the thrust deviation and thrust rate, respectively. The drift vector field and input matrix are
\begin{equation}
\mathbf{f}_0(\mathbf{x})=
\begin{bmatrix}
\mathbf{v} \\[3pt]
-\,g\hat{\mathbf e}_3+\dfrac{mg+T}{m}\mathbf{R}(\boldsymbol{\eta})\hat{\mathbf e}_3\\[6pt]
\dot{\boldsymbol{\eta}} \\[3pt]
\mathbf{0}_{3} \\[3pt]
\dot T \\[3pt]
0
\end{bmatrix},
~
\mathbf{G}_0=
\begin{bmatrix}
\mathbf{0}_{9\times 1} & \mathbf{0}_{9\times 3}\\
\mathbf{0}_{3\times 1} & \mathbf{I}_3\\
0 & \mathbf{0}_{1\times 3}\\
1 & \mathbf{0}_{1\times 3}
\end{bmatrix}.
\end{equation}
Here, $g>0$ is the gravitational constant, $m>0$ is the vehicle mass, $\hat{\mathbf e}_3=[0~0~1]^\top$, and $\mathbf{R}(\boldsymbol{\eta})\in SO(3)$ is the rotation matrix from body to inertial coordinates.

Let $\mathbf r_I^*\in\mathbb R^3$ denote the desired hover position, and let $\mathbf r_d(t)$ be a smooth reference trajectory from $\mathbf r(0)$ to $\mathbf r_I^*$ with bounded derivatives up to fourth order. Define the tracking errors
\[
\mathbf e_r=\mathbf r-\mathbf r_d,~
\mathbf e_v=\dot{\mathbf{r}}- \dot{\mathbf{r}}_d,~
\mathbf e_a=\ddot{\mathbf r}-\ddot{\mathbf r}_d,~
\mathbf e_j=\dddot{\mathbf r}_d-\ddot{\mathbf r}_d,~
\]
and the external tracking-error state
\begin{equation}
\mathbf z=
\begin{bmatrix}
\mathbf e_r^\top &
\mathbf e_v^\top &
\mathbf e_a^\top &
\mathbf e_j^\top &
\psi &
\dot\psi
\end{bmatrix}^\top
\in\mathbb R^{14}.
\end{equation}
Using differential flatness and dynamic inversion, the external error dynamics can be written as
\begin{equation}
\dot{\mathbf z}
=
\mathbf A_{\mathrm{EXT}}\mathbf z
+
\mathbf B_{\mathrm{EXT}}\mathbf s
-
\mathbf B_{\mathrm{EXT}}\mathbf r_d^{(4)}(t),
\end{equation}
where $\mathbf s\in\mathbb R^4$ is the virtual input,
\[
\mathbf A_{\mathrm{EXT}}=
\begin{bmatrix}
\mathbf{0}_{3} & \mathbf{I}_{3} & \mathbf{0}_{3} & \mathbf{0}_{3} & \mathbf{0} & \mathbf{0}\\
\mathbf{0}_{3} & \mathbf{0}_{3} & \mathbf{I}_{3} & \mathbf{0}_{3} & \mathbf{0} & \mathbf{0}\\
\mathbf{0}_{3} & \mathbf{0}_{3} & \mathbf{0}_{3} & \mathbf{I}_{3} & \mathbf{0} & \mathbf{0}\\
\mathbf{0}_{3} & \mathbf{0}_{3} & \mathbf{0}_{3} & \mathbf{0}_{3} & \mathbf{0} & \mathbf{0}\\
\mathbf{0} & \mathbf{0} & \mathbf{0} & \mathbf{0} & 0 & 1\\
\mathbf{0} & \mathbf{0} & \mathbf{0} & \mathbf{0} & 0 & 0
\end{bmatrix},
~
\mathbf B_{\mathrm{EXT}}=
\begin{bmatrix}
\mathbf{0} & 0\\
\mathbf{0} & 0\\
\mathbf{0} & 0\\
\mathbf{I}_3 & 0\\
\mathbf{0} & 0\\
\mathbf{0} & 1
\end{bmatrix}.
\]
The first three inputs of $\mathbf s$ act on the translational snap dynamics, and the fourth input acts on the yaw acceleration dynamics.

\subsection{Stability and Forward Invariance}
We consider a family of feedback laws parameterized by
\[
\mathbf k=\left[k_1~\cdots~k_{14}\right]^T\in\mathcal K\subset\mathbb R^{14},
\qquad
k_i\in[k_{i,\min},k_{i,\max}],
\]
for $i=1,\cdots,14$, and choose
\begin{equation}
\mathbf s=-\mathbf K\mathbf z,
\end{equation}
where $\mathbf K\in\mathbb R^{4\times 14}$ is constructed from $\mathbf k$. The resulting closed-loop dynamics are
\begin{equation}
\dot{\mathbf z}
=
\mathbf A_{\mathrm{cl}}(\mathbf k)\mathbf z
-
\mathbf B_{\mathrm{EXT}}\mathbf r_d^{(4)}(t),
\qquad
\mathbf A_{\mathrm{cl}}(\mathbf k)
\triangleq
\mathbf A_{\mathrm{EXT}}-\mathbf B_{\mathrm{EXT}}\mathbf K.
\end{equation}

\begin{theorem}
Consider the finite family of closed-loop error dynamics
\begin{equation}
\dot{\mathbf z}
=
\mathbf A_i \mathbf z
-
\mathbf B_{\mathrm{EXT}}\mathbf r_d^{(4)}(t),
\qquad i\in\{1,\ldots,N\},
\label{eq:switched_error_dynamics}
\end{equation}
where
$
\mathbf A_i
=
\mathbf A_{\mathrm{EXT}}
-
\mathbf B_{\mathrm{EXT}}\mathbf K_i,
$
and the finite controller library is
\[
\mathcal A
=
\{\mathbf K_1,\ldots,\mathbf K_N\}.
\]
Assume that there exist matrices
\[
\mathbf P=\mathbf P^\top\succ 0,
\qquad
\mathbf Q=\mathbf Q^\top\succ 0,
\]
such that
\begin{equation}
\mathbf A_i^\top \mathbf P
+
\mathbf P\mathbf A_i
\preceq
-\mathbf Q,
\qquad
\forall i\in\{1,\ldots,N\},
\label{eq:common_lyapunov_lmi}
\end{equation}
and that
\begin{equation}
\|\mathbf r_d^{(4)}(t)\|
\le
\bar r_4,
\qquad
\forall t\ge 0.
\label{eq:reference_bound}
\end{equation}
Then the switched closed-loop system is uniformly input-to-state
stable with respect to $\mathbf r_d^{(4)}(t)$ under arbitrary
switching among the controllers in $\mathcal A$. In particular, the common quadratic Lyapunov function
\begin{equation}
V(\mathbf z)
=
\mathbf z^\top\mathbf P\mathbf z
\label{eq:common_V}
\end{equation}
satisfies, for every admissible controller $\mathbf K_i$,
\begin{equation}
\dot V(\mathbf z)
\le
-\alpha\|\mathbf z\|^2
+
\beta\|\mathbf r_d^{(4)}(t)\|^2,
\label{eq:common_dissipation}
\end{equation}
where, for any
$\varepsilon\in(0,\lambda_{\min}(\mathbf Q))$,
\begin{equation}
\alpha
=
\lambda_{\min}(\mathbf Q)-\varepsilon>0,
\qquad
\beta
=
\frac{\|\mathbf P\mathbf B_{\mathrm{EXT}}\|^2}{\varepsilon}.
\label{eq:alpha_beta}
\end{equation}
Consequently, the same dissipation inequality holds independently
of the switching sequence.
\end{theorem}

\begin{proof}
Consider the common quadratic Lyapunov function
\[
V(\mathbf z)=\mathbf z^\top\mathbf P\mathbf z.
\]
Because $\mathbf P\succ0$,
\begin{equation}
\lambda_{\min}(\mathbf P)\|\mathbf z\|^2
\le
V(\mathbf z)
\le
\lambda_{\max}(\mathbf P)\|\mathbf z\|^2.
\label{eq:V_bounds}
\end{equation}
For any active controller $\mathbf K_i\in\mathcal A$,
differentiating $V$ along \eqref{eq:switched_error_dynamics}
gives
\begin{align}
\dot V
&=
\mathbf z^\top
\left(
\mathbf A_i^\top\mathbf P
+
\mathbf P\mathbf A_i
\right)
\mathbf z
-
2\mathbf z^\top
\mathbf P\mathbf B_{\mathrm{EXT}}
\mathbf r_d^{(4)}(t).
\end{align}
Using \eqref{eq:common_lyapunov_lmi},
\begin{align}
\dot V
&\le
-\mathbf z^\top\mathbf Q\mathbf z
+
2
\|\mathbf P\mathbf B_{\mathrm{EXT}}\|
\|\mathbf z\|
\|\mathbf r_d^{(4)}(t)\|\\
&\le
-\lambda_{\min}(\mathbf Q)\|\mathbf z\|^2
+
2
\|\mathbf P\mathbf B_{\mathrm{EXT}}\|
\|\mathbf z\|
\|\mathbf r_d^{(4)}(t)\|.
\end{align}
Applying Young's inequality,
\[
2ab
\le
\varepsilon a^2
+
\frac{1}{\varepsilon}b^2,
\qquad \varepsilon>0,
\]
with
\[
a=\|\mathbf z\|,
\qquad
b=
\|\mathbf P\mathbf B_{\mathrm{EXT}}\|
\|\mathbf r_d^{(4)}(t)\|,
\]
yields
\begin{align}
\dot V
\le
-&
\left(
\lambda_{\min}(\mathbf Q)-\varepsilon
\right)\|\mathbf z\|^2
\nonumber+
\frac{\|\mathbf P\mathbf B_{\mathrm{EXT}}\|^2}
{\varepsilon}
\|\mathbf r_d^{(4)}(t)\|^2.
\end{align}
Choosing
\[
0<\varepsilon<\lambda_{\min}(\mathbf Q)
\]
and defining $\alpha$ and $\beta$ as in
\eqref{eq:alpha_beta} gives
\[
\dot V
\le
-\alpha\|\mathbf z\|^2
+
\beta\|\mathbf r_d^{(4)}(t)\|^2.
\]

Importantly, $\mathbf P$, $\alpha$, and $\beta$ are common to
all controllers in $\mathcal A$. Therefore, the dissipation
inequality holds regardless of which admissible controller is
active and is preserved under arbitrary switching among the
controllers. Using
\[
\|\mathbf z\|^2
\ge
\frac{V(\mathbf z)}{\lambda_{\max}(\mathbf P)},
\]
we further obtain
\begin{equation}
\dot V
\le
-
\frac{\alpha}{\lambda_{\max}(\mathbf P)}V
+
\beta\|\mathbf r_d^{(4)}(t)\|^2.
\label{eq:V_ISS}
\end{equation}
Hence, the switched system is uniformly input-to-state stable
with respect to $\mathbf r_d^{(4)}(t)$ under arbitrary switching.
For the bounded input in \eqref{eq:reference_bound},
\eqref{eq:V_ISS} also establishes uniform ultimate boundedness.
\end{proof}
Theorem~1 establishes that the common Lyapunov certificate
is independent of the switching policy. Because the same
$\mathbf P$ certifies every controller in $\mathcal A$, the ISS
bound holds under arbitrary switching, including the
$\epsilon$-greedy action sequences encountered during DQN
training. Forward invariance of a prescribed admissible set
then follows by selecting an appropriate common Lyapunov
sublevel set, as established next.

\begin{remark}
\label{rem:cqlf_tradeoff}
Requiring a common $\mathbf P\succ0$ is more restrictive than
using mode-dependent Lyapunov functions, but it provides a
key advantage: the certificate remains valid under arbitrary
switching. Thus, RL may change controllers at any decision
instant without imposing dwell-time or hysteresis constraints.
The resulting conservatism is the price of obtaining a
policy-independent certificate.
\end{remark}
\noindent\textbf{Certified tracking-error admissibility:}
Let $\mathcal Z_{\mathrm{safe}}\subset\mathbb R^{14}$ denote the
prescribed admissible tracking-error set,
\begin{equation}
\mathcal{Z}_{\mathrm{safe}}
=
\left\{
\mathbf z\in\mathbb{R}^{14}:
|z_\ell|\le \bar z_\ell,\;
\ell=1,\ldots,14
\right\},
\label{eq:Z_safe}
\end{equation}
where $\bar z_\ell>0$ specifies the admissible bound on the
corresponding component of $\mathbf z$. Here, safety refers specifically to satisfaction of the
prescribed tracking-error bounds defining
$\mathcal Z_{\mathrm{safe}}$; additional physical constraints
not represented in $\mathcal Z_{\mathrm{safe}}$ are not implied
by this certificate. Define the largest
Lyapunov sublevel value whose ellipsoid is contained in
$\mathcal{Z}_{\mathrm{safe}}$ as
\begin{equation}
\rho_{\mathrm{safe}}
=
\min_{\ell=1,\ldots,14}
\frac{\bar z_\ell^2}
{\hat{\mathbf{e}}_\ell^\top\mathbf P^{-1}\hat{\mathbf{e}}_\ell},
\label{eq:rho_safe}
\end{equation}
where $\hat{\mathbf{e}}_\ell$ is the $\ell$th standard basis vector.
Then
\begin{equation}
\Omega_\rho
=
\left\{
\mathbf z\in\mathbb{R}^{14}:
V(\mathbf z)\le\rho
\right\}
\subseteq
\mathcal{Z}_{\mathrm{safe}}
\end{equation}
for every $\rho\le\rho_{\mathrm{safe}}$.

\begin{corollary}
\label{cor:forward_invariance}
Under the conditions of Theorem~1, suppose that
\begin{equation}
\frac{\lambda_{\max}(\mathbf P)\beta}{\alpha}
\bar r_4^2
\le
\rho
\le
\rho_{\mathrm{safe}}.
\label{eq:rho_interval}
\end{equation}
Then $\Omega_\rho$ is forward invariant under every
admissible switching signal $\sigma(t)$ and satisfies
\begin{equation}
\Omega_\rho\subseteq\mathcal{Z}_{\mathrm{safe}}.
\end{equation}
Consequently,
\begin{equation}
\mathbf z(0)\in\Omega_\rho
\quad\Longrightarrow\quad
\mathbf z(t)\in\mathcal{Z}_{\mathrm{safe}},
\qquad \forall t\ge0,
\end{equation}
under any switching policy whose actions belong to
$\mathcal A$. For the MDP formulation in Section~II, we therefore identify the
certified state space as
\begin{equation}
\mathcal X := \Omega_\rho .
\label{eq:X_certified}
\end{equation}
\end{corollary}

\begin{proof}
On the boundary $V(\mathbf z)=\rho$,
\begin{equation}
\|\mathbf z\|^2
\ge
\frac{\rho}{\lambda_{\max}(\mathbf P)}.
\end{equation}
Using the dissipation inequality in Theorem~1 and
$\|\mathbf r_d^{(4)}(t)\|\le\bar r_4$ gives
\begin{equation}
\dot V
\le
-\frac{\alpha}{\lambda_{\max}(\mathbf P)}\rho
+
\beta\bar r_4^2
\le 0,
\end{equation}
where the last inequality follows from the lower bound in
\eqref{eq:rho_interval}. Hence, the vector field of every
admissible mode points inward or is tangent to the boundary
of $\Omega_\rho$.

At every switching instant $t_j$,
\begin{equation}
V(t_j^+)=V(t_j^-),
\end{equation}
because the state is continuous and the same Lyapunov
function is shared by all modes. Therefore, neither continuous
evolution nor switching can drive the state outside
$\Omega_\rho$, establishing forward invariance under arbitrary
switching. Finally, the upper bound
$\rho\le\rho_{\mathrm{safe}}$ guarantees
$\Omega_\rho\subseteq\mathcal{Z}_{\mathrm{safe}}$ by
\eqref{eq:rho_safe}. Thus, every trajectory initialized in $\Omega_\rho$
remains in the prescribed tracking-error admissible set
$\mathcal Z_{\mathrm{safe}}$ for all $t\geq0$.
\end{proof}

\subsection{Nonlinear Closed-loop Dynamics}
Through dynamic inversion, the external input $\mathbf{s}$ is mapped to the physical control input $\mathbf{u}$ via
\[
\mathbf{s} = \mathbf{M}(\mathbf{x})\mathbf{u} + \mathbf{n}(\mathbf{x}),
\]
so that
\[
\mathbf{u} = \mathbf{M}^{-1}(\mathbf{x})\bigl(\mathbf{s} - \mathbf{n}(\mathbf{x})\bigr).
\]
Substituting this into the control-affine dynamics yields
\begin{equation}\label{dyn8}
\dot{\mathbf{x}} = \mathbf{f}(\mathbf{x}) + \mathbf{G}(\mathbf{x})\mathbf{k},
\end{equation}
which defines the nonlinear closed-loop system parameterized by $\mathbf{k}$.
 \noindent\textit{Inversion-domain assumption.} For nominal trajectories,
  $\sigma_{\min}(\mathbf M(\mathbf x))\ge\underline\sigma>0$, the Euler map
  stays nonsingular, and $\mathbf u=\mathbf h(\mathbf x,\mathbf k^{(i)})$
  does not saturate. Under these conditions and exact nominal-model matching,
  inversion realizes Theorem~1's external dynamics, so the result transfers
  through $\mathbf z(\mathbf x)$. Delay, mismatch, saturation, and singularity
  are outside this implication; their MuJoCo tests are empirical.

% Define the admissible set
% \[
% \mathcal{X}=\{x:\|\mathbf{z}(x)\|\le \delta\}.
% \]

% \subsection{Control-Oriented Discrete-Time Dynamics}
\begin{remark}
\label{rem:discrete_invariance}
At each decision instant $t_k$, the policy selects a controller
$\mathbf K_i\in\mathcal A$, which is held constant over
$[t_k,t_{k+1})$. Thus, the continuous-time dynamics
\eqref{eq:switched_error_dynamics} induce the discrete transition
\begin{equation}
\mathbf z_{k+1}
=
F(\mathbf z_k,a_k),
\qquad a_k\in\mathcal A,
\label{eq:certified_discrete_transition}
\end{equation}
where $\mathbf z_k=\mathbf z(t_k)$. With the certified MDP state
space $\mathcal X:=\Omega_\rho$, Corollary~\ref{cor:forward_invariance}
implies
\begin{equation}
\mathbf z_k\in\mathcal X
\quad\Longrightarrow\quad
\mathbf z_{k+1}\in\mathcal X,
\qquad \forall a_k\in\mathcal A.
\label{eq:discrete_invariance}
\end{equation}
Hence, every admissible action preserves the certified state space
at the RL decision instants.
\end{remark}
\section{DQN-based Safe Learning}
\label{DQN-based Safe Learning}

Given the certified action set $\mathcal A$, policy optimization is
formulated as a discrete-action RL problem. For the quadcopter case
study, the MDP state is the external tracking-error state
$\mathbf z_k=\mathbf z(t_k)\in\mathcal X=\Omega_\rho$. The stage
reward is
\begin{equation}
\begin{split}
r_k
=&-
\Big(
w_r \|\mathbf e_r\|^2
+ w_v \|\mathbf e_v\|^2
+ w_\eta \|\boldsymbol{\eta}\|^2
+ w_\omega \|\boldsymbol{\omega}\|^2
\Big)\\
&-w_u \|\mathbf u\|^2
-w_s\mathbf 1\{a_k\neq a_{k-1}\}.
\end{split}
\label{eq:reward}
\end{equation}
The terms associated with $\mathbf e_r$ and $\mathbf e_v$ penalize
tracking error and residual translational motion, while the attitude,
angular-rate, and control-effort terms discourage aggressive vehicle
motion and excessive actuation. The switching penalty discourages
frequent changes between controllers. Importantly, no safety penalty
is included in the reward: forward invariance is guaranteed by the
certified action set independently of the reward design and learned
policy.

The action-value function is approximated by a neural network
$Q(\mathbf z,a;\theta)$, where $\theta$ denotes the trainable
parameters. At decision instant $k$, the action is selected according
to an $\epsilon$-greedy policy,
\begin{equation}
a_k=
\begin{cases}
\text{a random action in }\mathcal A,
& \text{with probability }\epsilon,\\[1mm]
\displaystyle
\arg\max_{a\in\mathcal A}Q(\mathbf z_k,a;\theta),
& \text{with probability }1-\epsilon.
\end{cases}
\label{eq:epsilon_greedy}
\end{equation}

\begin{figure*}[!t]
      \centering
      \includegraphics[width=\textwidth]
          {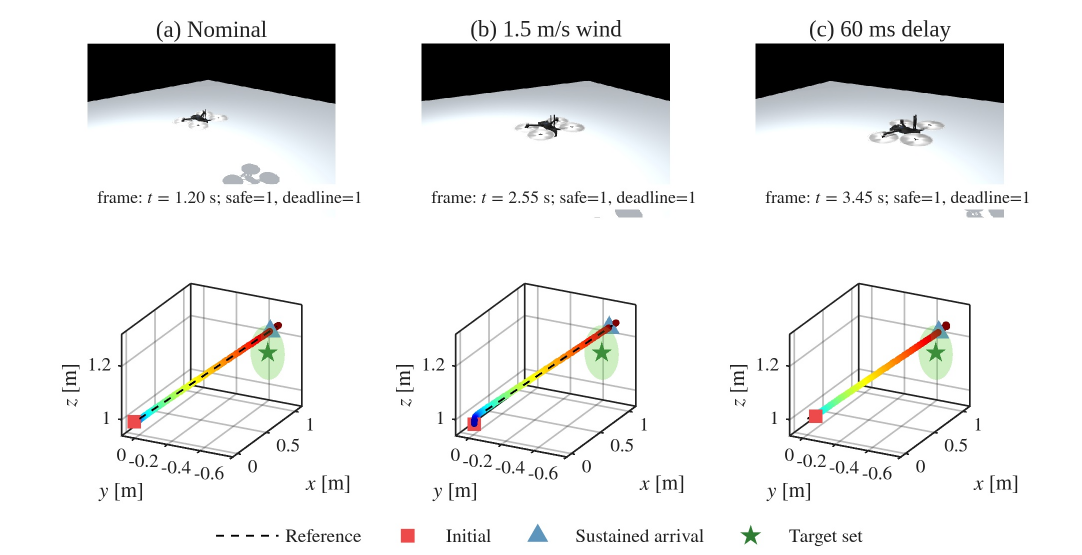}
      \vspace{-1.5mm}
      \caption{Matched MuJoCo quadrotor executions under nominal conditions,
      1.5-m/s wind, and 60-ms sensing delay. Each fixed-world rendering and
      3-D trajectory pair comes from the same proposed-DQN rollout. Color
      advances with time from launch to sustained arrival; dashed curves show
      the reference, red squares the initial position, blue triangles sustained
      arrival, and green spheres and stars the terminal set and its center.}
      \label{fig:cdc81_demo}
  \end{figure*}

\section{Experimental Evaluation}
\label{subsec:sim_setup}

We evaluate nominal performance, scheduling, and robustness.
Four questions guide the experiments: \textbf{Q1}, common certification over
the initial set; \textbf{Q2}, benefits of state-dependent scheduling;
\textbf{Q3}, the distinction between common certification and individual-mode
stability; and \textbf{Q4}, behavior beyond the nominal model.
Figure~\ref{fig:cdc81_demo} previews representative proposed-DQN executions
under the nominal, wind, and sensing-delay conditions examined below.

\subsection{MuJoCo plants, task, and frozen protocol}

The common evaluation setting addresses \textbf{Q1} by testing the controller
library's certificate over the prescribed initial set. The primary experiment uses a nonlinear MuJoCo rigid body whose nominal mass
and inertia, $1.5$~kg and $\operatorname{diag}(.02,.02,.04)$~kg,m$^2$, match
the controller model. RK4 integration uses a 2-ms step, and the controller is
updated every 20~ms. Each 10-s episode tracks a ninth-order smoothstep from
$(0,0,1)$ to $(1,-.5,1.25)$~m during the first 5~s and then holds the terminal
position. Its first four derivatives vanish at both endpoints. Initial errors
satisfy $|e_{r,x}|,|e_{r,y}|\leq.05$~m, $|e_{r,z}|\leq.03$~m, and
$|e_{v,i}|\leq.04$~m/s, with all higher-order initial errors set to zero.

The task requires the vehicle to enter
$|r-r_f|\leq.10$~m and $|v|\leq.20$~m/s by 3.66~s and subsequently
maintain precision hover. Empirical physical safety requires altitude at least
$.15$~m, lateral distance at most $6$~m, tilt at most $1.05$~rad, and finite
computed states. We report this metric, deadline satisfaction,
sustained-arrival time, conditional post-arrival position RMSE, actuator
saturation, maximum $V/\rho$, and switching count. The proposed library combines translational root scales
${.8,1,1.2}$ with three yaw-root pairs, yielding 81 gain vectors. The
residual-adjusted common-certificate margin is $q=.00606588$, with sufficient
sublevel value $\rho=2.27167$. Evaluation of all 64 corners of the initial-error
box gives $\max V(z_0)/\rho=.001141<1$, confirming that the prescribed initial
set lies inside the certified sublevel set. The scheduler input contains the 14 error states, an 81-dimensional one-hot
encoding of the previous action, and normalized reference progress. Each
selected library action is retained for five controller updates (0.10~s).

We also test the certificate near its boundary: at each
$V(z_0)/\rho\in\{.05,.25,.50,.75,.90\}$, 20 physically consistent states
are constructed by jointly scaling position, velocity, attitude, angular rate,
and thrust states. The median fixed controller, five DQN checkpoints, and
random certified switching give 700 nominal rollouts. All remain physically
safe with $\max_tV/\rho\le1$ and zero saturation; sampled diagnostics give
$\sigma_{\min}(\mathbf M)\ge.402$, condition number $\le28.92$, and tilt
$\le.633$~rad. This is not a global numerical proof of the inversion-domain
assumption.
Table~\ref{tab:arrival_training} summarizes the frozen learning protocol.
Training, checkpoint selection, and all hyperparameter choices were completed
before evaluation on the 40 paired test scenarios.

\begin{table}[t]
\centering\footnotesize
\caption{Frozen learning and evaluation protocol for the arrival task.}
\label{tab:arrival_training}
\setlength{\tabcolsep}{3pt}
\begin{tabularx}{\columnwidth}{@{}lX@{}}
\toprule
Item & Setting\\
\midrule
Observation & $14+81+1=96$ variables\\
Network/output & ReLU $96$--$256$--$256$--$81$\\
Initialization & 1500 cloning minibatches; 20 teacher episodes\\
Fine tuning & five seeds; 60,000 Double-DQN interactions each\\
Optimization & Adam $10^{-3}$; $\gamma=.99$; batch 256; gradient norm 10\\
Replay/target & capacity 200k; target update every 1000 steps\\
Exploration & linear $\epsilon:.10\rightarrow.01$; 5\% teacher mixture\\
Selection & safety, deadline rate, hover RMSE, then constrained cost\\
Splits & development 304000--305039; validation 306000--306019; test 307000--307039\\
\bottomrule
\end{tabularx}
\end{table}

\subsection{Main comparison: safety, performance, and scheduling}

The main comparison addresses \textbf{Q2} and \textbf{Q3} by separating the
effects of state-dependent scheduling and common certification. The fixed baseline
uses the median-library action 40, selected before evaluation because it is the
central row of the $81$-action construction: all translational root scales are
one and the yaw roots are $(2,6)$. The proposed DQN selects among all 81
certified actions. Because both policies use the same common-certified library,
their comparison isolates the benefit of state-dependent scheduling relative to
a pre-specified, non-adaptive nominal tuning.

To examine the role of the common certificate, an otherwise matched DQN
ablation replaces only the gain library. Its root scales are
${.6,1,1.6}$, and its independently selected fast and steady actions are 36
and 80. All 81 ablation modes are Hurwitz (largest spectral abscissa $-.6$),
but the same trace-normalized common-$P$ program returns margin
$-1.56\times10^{-5}$. Thus, individual controller stability alone does not provide
the positive common quadratic margin established for the proposed library.
This calculation does not rule out other possible switching certificates.

All three policies are safe and meet the deadline in every nominal test
rollout. Relative to the median/nominal fixed
controller, the proposed DQN advances sustained arrival by .0091~s and reduces
hover RMSE from $.013326$ to $.010937$~m, a 17.93\% reduction (paired mean
difference $.002389$~m; 95\% interval $[.002199,.002579]$~m;
$p=6.99\times10^{-26}$). Across the five frozen checkpoints, the scheduling
policy exhibits a measurable performance--switching trade-off. Thus, the
results demonstrate a repeatable advantage of certified state-dependent
scheduling over a representative non-adaptive nominal tuning.

To quantify the effect of switching at execution time, we additionally
evaluated the frozen policies on 40 previously unused paired scenarios
(seeds 313000--313039), without retraining or checkpoint selection.
For each scenario, the DQN metric was averaged over the five frozen
checkpoints before paired comparison with the fixed controller.
Figure~\ref{fig:arrival_dwell_tradeoff} shows that the .10-s DQN reduces hover RMSE by
18.65\% (mean difference $.002485$~m; paired 95\% interval
$[.002243,.002728]$~m; $p=1.10\times10^{-22}$). Increasing the minimum hold
to .20~s reduces mean switching by 46\% while retaining a 16.59\% RMSE
reduction (mean difference $.002210$~m; paired 95\% interval
$[.001948,.002473]$~m; $p=1.11\times10^{-19}$). At .50~s, the further
reduction in switching is accompanied by a substantially smaller tracking
margin. Thus, the hold interval provides an explicit execution-level
tracking--switching trade-off; the common certificate itself does not rely on
dwell time.

\begin{figure}[!t]
\centering
\includegraphics[width=.98\columnwidth]{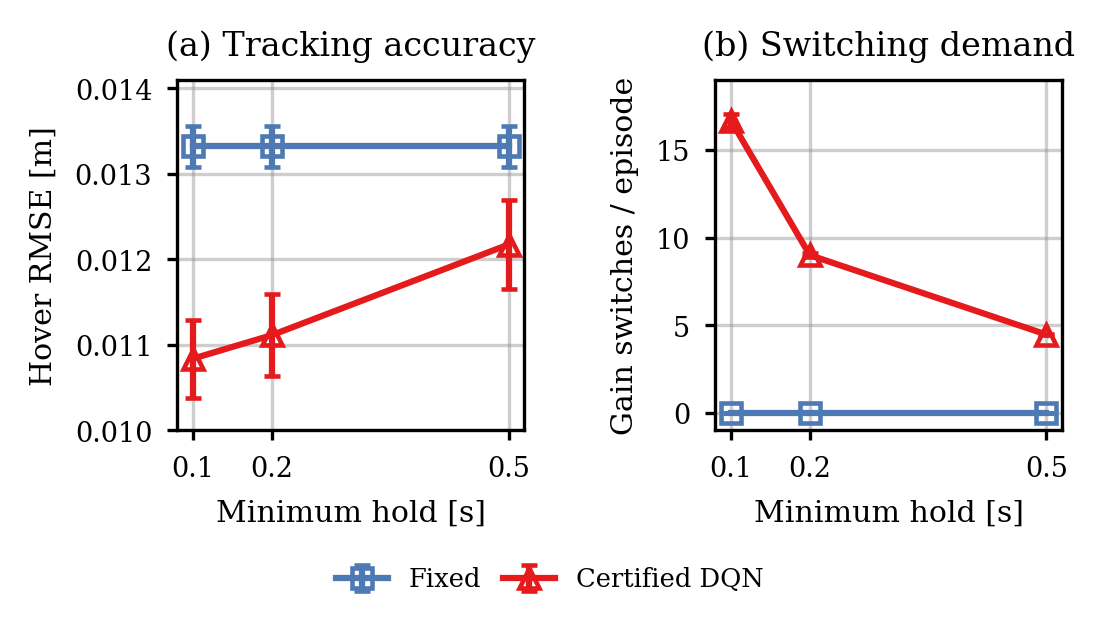}
\caption{Fresh-seed dwell-time ablation. Markers and error bars denote means
and paired-scenario 95\% confidence intervals, respectively; each DQN value
is first averaged over the five frozen checkpoints within each scenario. The
.20-s hold retains a tracking advantage over the median/nominal fixed controller
while reducing DQN switching relative to the .10-s hold.}
\label{fig:arrival_dwell_tradeoff}
\end{figure}

The uncertified-library DQN also succeeds nominally. Its nominal outcome does
not supply the missing common-certificate result: among the two tested
libraries, only the proposed library carries the stated arbitrary-switching
guarantee for the ideal model. The nonideal tests below evaluate whether this
analytical distinction is accompanied by empirical separation.

\subsection{Robustness and empirical applicability boundary}

The robustness tests address \textbf{Q4} beyond the nominal certificate. The robustness evaluation reuses 40 paired scenario seeds under nominal
operation, wind from 0 to 3~m/s, plant mismatch, sensing noise, 60-ms delay,
and a joint condition. Plant mismatch perturbs mass, inertia, and actuator
effectiveness; the joint condition additionally combines 1.5-m/s wind, noise,
and 60-ms delay. Each deterministic baseline contributes 40 rollouts per
condition, while each DQN row aggregates five frozen checkpoints for 200
rollouts. Safety and deadline rates use all rollouts, while post-arrival RMSE
is reported only for trajectories that achieve sustained arrival.

\begin{table*}[!t]
\centering
\caption{\textbf{Held-out nominal and nonideal safety/arrival performance.}
Arrows indicate the preferred direction: higher safety and deadline rates
($\uparrow$), and lower conditional post-arrival RMSE ($\downarrow$).}
\label{tab:arrival_robustness}
\renewcommand{\arraystretch}{1.08}
\setlength{\tabcolsep}{2.5pt}
\resizebox{\textwidth}{!}{%
\begin{tabular}{l ccc ccc ccc ccc ccc ccc}
\toprule
\multirow{2}{*}{\textbf{Method}} &
\multicolumn{3}{c}{\textbf{Nominal}} &
\multicolumn{3}{c}{\textbf{Wind 2 m/s}} &
\multicolumn{3}{c}{\textbf{Model mismatch}} &
\multicolumn{3}{c}{\textbf{Sensor noise}} &
\multicolumn{3}{c}{\textbf{Delay 60 ms}} &
\multicolumn{3}{c}{\textbf{Joint}}\\
\cmidrule(lr){2-4}\cmidrule(lr){5-7}\cmidrule(lr){8-10}
\cmidrule(lr){11-13}\cmidrule(lr){14-16}\cmidrule(lr){17-19}
& Safe & Dead. & RMSE & Safe & Dead. & RMSE & Safe & Dead. & RMSE &
Safe & Dead. & RMSE & Safe & Dead. & RMSE & Safe & Dead. & RMSE\\
\midrule
Median/nominal fixed &100&100&.01322&100&100&.12529&100&20.0&.03933&100&100&.01377&100&100&.01414&100&10.0&.08933\\
\textbf{Proposed DQN} &100&100&.01067&100&82.0&.09932&100&22.5&.04978&100&100&.01140&\textbf{100}&\textbf{100}&.01275&\textbf{100}&12.5&.08327\\
Uncertified DQN &100&100&.01213&100&100&.12527&100&6.5&.03593&100&100&.01264&2&81.5&.32976&0&0&--\\
\bottomrule
\end{tabular}%
}

\begin{minipage}{.98\textwidth}
\footnotesize\textit{Notes.} Safe and Dead. denote percentages. RMSE is in
meters and is conditional on sustained arrival; ``--'' indicates that no
rollout produced a post-arrival segment. DQN entries aggregate five frozen
checkpoints on each of the same 40 scenario seeds.
\end{minipage}
\end{table*}
Table~\ref{tab:arrival_robustness} shows that delay and combined perturbations
produce the clearest separation. With 60-ms delay, both policies using the
common-certified library remain safe in every rollout,
whereas the uncertified DQN is safe in only 2\% and meets the deadline in
81.5\%. Under the joint condition, the median fixed and proposed DQN remain
safe in every run, whereas the uncertified DQN has zero safety and deadline
satisfaction. These results are empirical and do not extend the ideal-model
certificate to delayed or mismatched dynamics.

Wind does not trigger physical-safety termination through 3~m/s, although
deadline performance degrades before that boundary.
Figure~\ref{fig:arrival_robustness}(a) shows the 0--2~m/s range, while panels
(b) and (c) compare safety and joint completion across nonideal conditions.
Figure~\ref{fig:arrival_boundary} separately reports near-boundary nominal
Lyapunov histories. At 2.5~m/s,
deadline satisfaction degrades for every policy, while no method meets the
deadline at 3~m/s. We therefore treat the higher wind levels as an empirical
task-performance boundary. For all nonideal tests, $V/\rho$ is used only as a
diagnostic computed from the nominal certificate.

\begin{figure}[!t]
\centering
\includegraphics[width=.98\columnwidth]{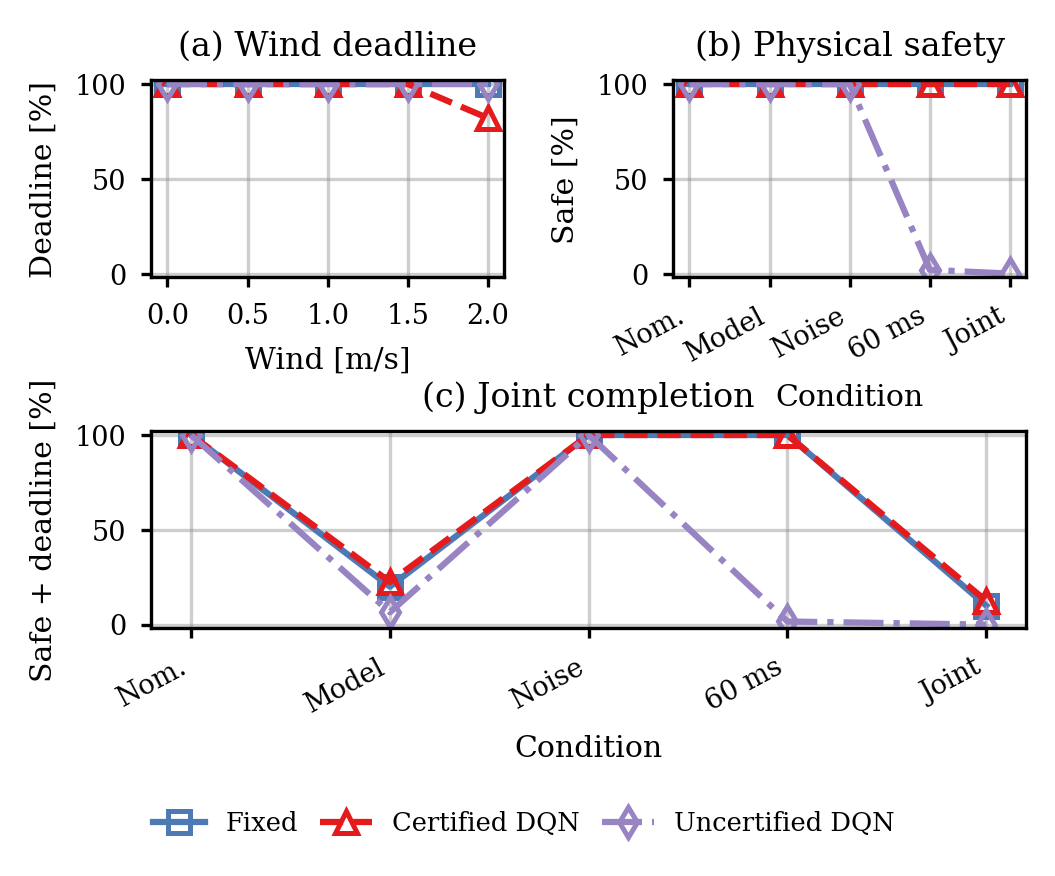}
\caption{Empirical robustness beyond the nominal certificate. (a) Deadline
satisfaction under 0--2~m/s wind. (b) Physical safety and (c) joint
safe-and-deadline completion under the principal nonideal conditions.}
\label{fig:arrival_robustness}
\end{figure}

\begin{figure}[!t]
\centering
\includegraphics[width=.98\columnwidth]{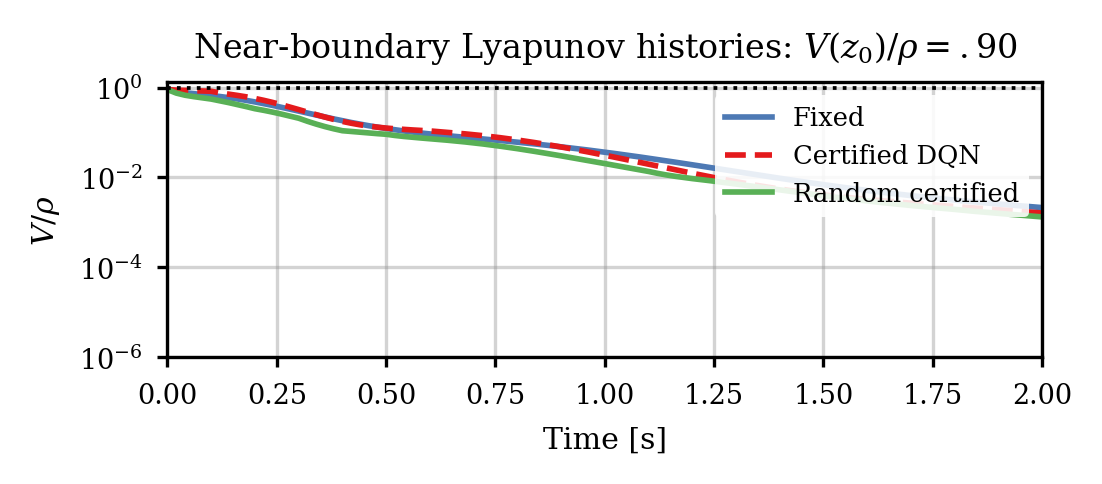}
\caption{Nominal near-boundary Lyapunov histories from one common,
physically consistent initial state with $V(z_0)/\rho=.90$. Fixed, learned,
and random certified switching remain below the certificate boundary
$V/\rho=1$ while exhibiting distinguishable decay transients. This rollout
illustrates the policy-class result; the 700-rollout containment summary is
reported in the text.}
\label{fig:arrival_boundary}
\end{figure}

A separate 100-ms-delay stress test marks an empirical failure boundary. No
policy remains physically safe, and the rollouts exhibit saturation and large
excursions in the nominal $V/\rho$ diagnostic.

\section{Conclusion}\label{Conclusion}

This paper presented an RL framework for safe gain scheduling based on
forward-invariance-induced action-space design. By constructing
a finite controller library sharing a common Lyapunov
certificate, the framework establishes, under its nominal assumptions,
forward invariance of a prescribed admissible set under arbitrary switching. Safety is
therefore a property of the admissible policy class rather than
of the learned policy itself: learning optimizes performance
within the certified class without requiring runtime safety
filters, action projection, or dwell-time constraints. The quadcopter case study demonstrated this separation between
certification and policy optimization. The DQN performed
state-dependent gain scheduling over the certified controller
library while preserving the nominal invariance guarantee, and
nonlinear MuJoCo simulations empirically examined robustness
beyond the assumptions of the certificate. For the considered
arrival task, the learned scheduler improved tracking relative
to the pre-specified nominal fixed controller while the
dwell-time ablation exposed its explicit execution-level
switching trade-off. The principal benefit of the proposed
formulation is the ability to optimize policies within a
policy-independent safety certificate. The common quadratic certificate is sufficient but potentially
conservative, since individually stabilizing controllers may
fail to admit a common $\mathbf P$. Future work will investigate
less conservative invariance certificates, larger or continuous
action spaces, and extensions to autonomous driving with
safety-certified steering and braking under lane-keeping,
obstacle-avoidance, and vehicle-stability constraints.

\providecommand{\BIBdecl}{}
\renewcommand{\BIBdecl}{%
  \footnotesize
  \setlength{\itemsep}{0pt}%
  \setlength{\parsep}{0pt}%
  \setlength{\parskip}{0pt}%
}
\bibliographystyle{IEEEtran}
\bibliography{CDC-Ref}

\end{document}